\documentclass[11pt]{article}
\usepackage{fullpage}
\usepackage{algorithmic}
\usepackage{epsfig}
\usepackage{epstopdf}
\usepackage{graphicx}
\usepackage{latexsym}
\usepackage{amsmath}
\usepackage{amsfonts}
\usepackage{amssymb}
\usepackage{mathrsfs}
\usepackage{pifont}
\usepackage{yhmath}
\usepackage{undertilde}
\usepackage{bbm}
\usepackage{eufrak}
\usepackage{amsthm}
\usepackage{booktabs}
\usepackage[usenames,dvipsnames]{color}
\usepackage{stmaryrd}
\usepackage{wasysym}
\usepackage{bbm}
\usepackage{array}
\usepackage[ruled]{algorithm2e}

\usepackage{bigstrut,multirow,rotating}

\usepackage{tabularx, environ}

\makeatletter

\newcolumntype{\expand}{}
\long\@namedef{NC@rewrite@\string\expand}{\expandafter\NC@find}

\NewEnviron{problem}[2][]{%
  \def\problem@arg{#1}%
  \def\problem@framed{framed}%
  \def\problem@lined{lined}%
  \def\problem@doublelined{doublelined}%
  \ifx\problem@arg\@empty%
    \def\problem@hline{}%
  \else%
    \ifx\problem@arg\problem@doublelined%
      \def\problem@hline{\hline\hline}%
    \else%
      \def\problem@hline{\hline}%
    \fi%
  \fi%
  \ifx\problem@arg\problem@framed%
    \def\problem@tablelayout{|>{\bfseries}lX|c}%
    \def\problem@title{\multicolumn{2}{|l|}{%
        \raisebox{-\fboxsep}{\textsc{\Large #2}}%
      }}%
  \else
    \def\problem@tablelayout{>{\bfseries}lXc}%
    \def\problem@title{\multicolumn{2}{l}{%
        \raisebox{-\fboxsep}{\textsc{\Large #2}}%
      }}%
  \fi%
  \bigskip\par\noindent%
  \begin{tabularx}{\textwidth}{\expand\problem@tablelayout}%
    \problem@hline%
    \problem@title\\[2\fboxsep]%
    \BODY\\\problem@hline%
  \end{tabularx}%
  \medskip\par%
}
\makeatother

\usepackage{bigstrut,multirow,rotating}

\newtheorem{theorem}{Theorem}[section]
\newtheorem{lemma}[theorem]{Lemma}

\newtheorem{mechanism}[theorem]{Mechanism}

\theoremstyle{remark}
\newtheorem{remark}[theorem]{Remark}
\newtheorem{example}[theorem]{Example}

\makeatletter
\newcommand\figcaption{\def\@captype{figure}\caption}
\newcommand\tabcaption{\def\@captype{table}\caption}
\makeatother

\DeclareMathAlphabet{\mathpzc}{OT1}{pzc}{m}{it}

\begin{document}
\newcounter{my}
\newenvironment{mylabel}
{
\begin{list}{(\roman{my})}{
\setlength{\parsep}{-1mm}
\setlength{\labelwidth}{8mm}
\usecounter{my}}
}{\end{list}}

\newcounter{my2}
\newenvironment{mylabel2}
{
\begin{list}{(\alph{my2})}{
\setlength{\parsep}{-0mm} \setlength{\labelwidth}{8mm}
\setlength{\leftmargin}{3mm}
\usecounter{my2}}
}{\end{list}}

\newcounter{my3}
\newenvironment{mylabel3}
{
\begin{list}{(\alph{my3})}{
\setlength{\parsep}{-1mm}
\setlength{\labelwidth}{8mm}
\setlength{\leftmargin}{10mm}
\usecounter{my3}}
}{\end{list}}

\title{\bf Mechanism Design for Facility Location Games\\ with Candidate Locations}

\date{}
\maketitle
\vspace{-3em}
\begin{center}

\author{Zhongzheng Tang$^{1}$\quad Chenhao Wang$^{2}$\quad Mengqi Zhang$^{3,4}$\quad Yingchao Zhao$^{5}$\\
${}$\\
$^1$ School of Sciences, Beijing University of Posts and\\ Telecommunications, Beijing
100876, China\\
$^2$ Department of Computer Science and Engineering, University of\\ Nebraska-Lincoln, NE, United States\\
$^3$ Academy of Mathematics and Systems Science, Chinese Academy\\ of Sciences, Beijing 100190, China\\
$^4$ School of Mathematical Sciences, University of Chinese Academy\\ of Sciences, Beijing 100049, China\\
$^5$ Caritas Institute of Higher Education, HKSAR, China\\
\medskip
\{tangzhongzheng,wangch,mqzhang\}@amss.ac.cn, zhaoyingchao@gmail.com}

\end{center}
\vspace{1em}


\begin{abstract}
We study the facility location games with candidate locations from a mechanism design perspective. Suppose there are $n$ agents located {in} a metric space whose locations are their private information, and a group of candidate locations for building facilities. The authority plans to build some homogeneous facilities among these candidates to serve the agents, who {bears} 
a cost equal to the distance to the closest facility. The goal is to design mechanisms for minimizing the total/maximum cost among the agents. For the single-facility problem under the {maximum-cost} objective, we give a deterministic 3-approximation group strategy-proof mechanism, and prove that no deterministic (or randomized) strategy-proof {mechanism} can have an approximation ratio better than 3 (or 2). For the two-facility problem on a {line}, we give an anonymous deterministic group strategy-proof mechanism that is $(2n-3)$-approximation for the {total-cost objective}, and 3-approximation for the {maximum-cost objective}. We also provide (asymptotically) tight lower bounds on the approximation ratio.\\

\noindent{\bf Keywords:}~Facility location; Social choice; Mechanism design
\end{abstract}

\section{Introduction}

We consider a well-studied facility location problem of deciding where some public facilities should be built to {serve} a population of agents with their locations as private information. For example, a government needs to decide the locations of public supermarkets or hospitals. It is often modeled in a metric space or a network, where there are some agents (customers or citizens)
who may benefit by misreporting their locations. This manipulation can be problematic for a decision maker to find a system optimal solution, and {leads} to the mechanism design problem of providing (approximately) optimal solutions while also being strategy-proof (SP), i.e., no agent can be better off by misreporting their locations, regardless of what others report.

This setup, where the agents are located {in} a network that is represented as a contiguous graph, is initially studied by Schummer and Vohra \cite{schummer2002strategy}, and has many applications (e.g., traffic network). Alon \emph{et al.} \cite{alon2010strategyproof} give an example of telecommunications networks such as a local computer network or the Internet. In these cases, the agents are the network users or service providers, and the facility can be a filesharing server or a router. 
Interestingly, in computer networks, an agent's perceived network location can be easily manipulated, for example, by generating a false IP address or rerouting incoming and outgoing communication, etc. This explains the incentive of agents for misreporting, and thus a strategy-proof mechanism is necessary.

In the classic model \cite{procaccia2013approximate,lu2010asymptotically}, all points in the metric space or the network are feasible for building facilities. However, this is often impractical in many applications. For example, due to land use restrictions, the facilities can only be built in some feasible regions, while other lands are urban green space, residential buildings and office buildings, etc.  Therefore, we assume that there is a set of candidate locations, and study the facility location games with candidate locations in this paper.  


 We notice that our setting somewhat coincides a metric social choice problem \cite{feldman2016voting}, where the voters (agents) have their preferences over the candidates, and all participants are located {in} a metric space, represented as a point. The voters prefer candidates that are closer to them to the ones that are further away. The goal is to choose a candidate as a winner, such that the total distance to all voters is as small as possible. When the voters are required to report their locations, this problem is the same with the facility location game with candidate locations.

Our  setting  is sometimes referred to as the ``constrained facility location" problems \cite{sui2015approximately},
as the feasible locations for facilities are constrained. Sui and  Boutilier \cite{sui2015approximately} provide  possibility and
impossibility results with respect to (additive) approximate individual and group strategy-proofness, whereas do not consider the approximation ratios for system objectives.

\medskip\noindent\textbf{Our results.}

In this paper we study the problem of locating one or two facilities {in} 
a metric space, where there are $n$ agents and a set of feasible locations for building the facilities. For the {single-facility problem}, 
we consider the objective of minimizing the maximum cost among the agents, while the {social-cost} (i.e., the total cost of agents) objective has been well studied in \cite{feldman2016voting} as a voting process. We present a mechanism that deterministically selects the closest candidate location to {an} arbitrary dictator agent, and prove that it is group strategy-proof (GSP, no group of agents being better off by misreporting) and 3-approximation. In particular, when the space is a line, the mechanism that selects the closest candidate location to the leftmost agent is additionally anonymous, that is, the outcome is the same for all permutations of the agents' locations on the line. We {provide} 
a lower bound 3 for deterministic SP mechanisms, and 2 for randomized SP mechanisms; both lower bounds hold even on a line.

For the {two-facility problem} 
on a 
{line}, we present an anonymous GSP mechanism that deterministically selects two candidates closest to the leftmost and rightmost agents, respectively.  It is $(2n-3)$-approximation for the {social-cost} objective, and $3$-approximation for the {maximum-cost} objective. On the negative side, we prove that, for the {maximum-cost} objective, no deterministic (resp. randomized) strategy-proof mechanism can have an approximation ratio better than 3 (resp. 2).

Our results for deterministic mechanism on a line are summarized in Table \ref{tab:2} in bold font, where a ``$\star$" indicates that the upper bound holds for general metric spaces. All inapproximability results are obtained on a line, and thus hold for more general metric spaces.

\begin{table}[htbp]
  \centering
  \caption{Results for deterministic strategyproof mechanisms on a line.} \label{tab:2}
\smallskip
    \begin{tabular}{|p{2.3 cm}<{\centering}|p{3.718cm}<{\centering}|p{3.7cm}<{\centering}|}
    \hline
    \small Objective  &  \small {Social cost} &\small {Maximum cost} \bigstrut\\
    \hline
    \multirow{2}[2]{*}{\small {Single-facility}} 
    & \multirow{1}[1]{*}{\small ~~UB: $3^\star$ \cite{feldman2016voting} } & {\small ~~UB: $\mathbf{3^\star}$ } \bigstrut[t]\\
    & {\small LB: $3$ \cite{feldman2016voting} } & {\small LB: $\textbf{3}$} \bigstrut[b]\\
    \hline
    \multirow{2}[2]{*}{\small {Two-facility}} 
    & \multirow{1}[1]{*}{\small ~~~~UB: $\mathbf{2n-3}$ } & {\small ~UB: $\textbf{3}$  } \bigstrut[t]\\

                                    & {\small \hspace{1.7mm}~~~~LB: $n-2$ \cite{fotakis2014power} } & {\small \hspace{0.5mm}LB: $\textbf{3}$} \bigstrut[b]\\
    \hline
    \end{tabular}
\end{table}

\smallskip\noindent\textbf{Related work.}

A range of works on social choice study the constrained single-facility location games for the {social-cost} objective,  where agents can be placed anywhere, but only a subset of locations is valid for the facility. The \emph{random dictatorship} (RD) mechanism, which selects each candidate with probability equal to the fraction of agents who vote for it, obtains {an approximation ratio of} $3-2/n$, 
and this is tight for all strategyproof mechanisms \cite{feldman2016voting,meir2012algorithms}.  The upper bound holds for any metric {spaces}, whereas the lower
bound requires specific constructions on the $n$-dimensional binary cube. Anshelevich and Postl \cite{anshelevich2017randomized} show a smooth transition of the RD approximation ratio from $2-2/n$ to $3-2/n$ as the location of the facility becomes more constrained. Meir \cite{meir2018strategic} (Section 5.3) provides an overview of approximation results for the single-facility problem.

\vspace{-1.5mm}
\paragraph{Approximate mechanism design in the classic setting.} For the classic facility location games wherein the locations have no constraint, Procaccia and Tennenholtz \cite{procaccia2013approximate} first consider it from the perspective of approximate mechanism design. For single-facility location on a line, they give a ``median" mechanism that is GSP and optimal for minimizing the social cost. Under the {maximum-cost} objective, they provide a deterministic 2-approximation and a randomized 1.5 approximation GSP mechanisms; both bounds are best possible.  For two-facility location, they give a 2-approximation mechanism  that always places the facilities at the leftmost and the rightmost locations of agents. Fotakis and Tzamos \cite{fotakis2014power} characterize deterministic mechanisms  for the problem of locating two facilities on the line, and prove a lower bound of $n-2$. Randomized mechanisms are considered in \cite{lu2009tighter,lu2010asymptotically}.


\vspace{-1.5mm}
\paragraph{Characterizations.} Dokow \emph{et al.} \cite{dokow2012mechanism} study SP mechanisms for locating a facility {in} 
a discrete graph, where the agents are located on vertices of the graph, and the possible facility locations are exactly the vertices of the {graph}. They give a full characterization of SP mechanisms on lines and 
sufficiently large cycles. For continuous lines, the set of SP and onto  mechanisms has been characterized
as all generalized median voting schemes \cite{border1983straightforward,schummer2002strategy}. 
\vspace{-1.5mm}
\paragraph{Other settings.} There are many different settings for facility location games in recent years. 
Aziz \emph{et al.} \cite{Aziz2019TheCC} study the mechanism design problem where the public facility is capacity constrained, where the capacity constraints limit the number of agents who can benefit from the facility's services.
Chen \emph{et al.} \cite{chen2019truthful} study a dual-role game where each agent can allow a
facility to be opened at his place and he may strategically report his opening cost. By introducing payment, they characterize truthful mechanisms and provide approximate mechanisms. After that, Li \emph{et al.} \cite{libudgeted} study a model with payment under a budget constraint.
Kyropoulou \emph{et al.} \cite{kyropoulou2019mechanism} initiate the study of constrained heterogeneous facility location problems, wherein selfish agents can either like or dislike the facility and facilities can be located {in} 
a given feasible region of the Euclidean plane.  
Other works on heterogeneous facilities can be found in \cite{zou2015facility,serafino2015truthful,chen2018mechanism}.
Independently, Walsh \cite{walsh2020strategy} studies a similar but more general setting where facilities have a finite set of feasible subintervals, and different facilities can have different feasible regions (e.g., the left facility is in $[1,2]$ and the right facility is in $[3,4]$). He proves tight upper bound 3 for both social cost and maximum cost objectives. By contrast, for two or more facilities, no anonymous mechanism has bounded approximation ratio for either objective. Additionally, he studies the problems of maximizing utilitarian and egalitarian social welfare, by assuming each agent has a distance-discounted utility.

\section{Model}\label{model1}
Let $k$ be the number of facilities to {be built}. 
In an instance of facility location game with candidate locations, 
the agent set is $N=\{1,\ldots,n\}$, and each agent $i\in N$ has a private location $x_i\in S$ {in} 
a metric space $(S,d)$, where $d:S^2\rightarrow \mathbb R$ is the metric (distance function). We denote by $\mathbf x=(x_1,\ldots,x_n)$ the location profile of agents. {The set of $m$ candidate locations is $M\subseteq S$.} 
A deterministic mechanism $f$ takes the reported agents' location profile $\mathbf x$ as input, and {outputs} a facility location profile $\mathbf y=(y_1,\ldots,y_k)\in M^k$, that is, selecting $k$ candidates for building facilities. A randomized mechanism outputs a probability distribution over $M^k$. Given an outcome $\mathbf y$, the \emph{cost} of each agent $i\in N$ is the distance to the closest facility, i.e.,
$c_i(\mathbf y)=d(x_i,\mathbf y):=\min_{1\le j\le k}d(x_i,y_j).$

A mechanism $f$ is \emph{strategy-proof} (SP), if no agent $i\in N$ can decrease his cost by misreporting, regardless of the location profile $\mathbf x_{-i}$ of others, that is, for any $x_i'\in S$,
{$c_i(f(x_i,\mathbf x_{-i}))\le c_i(f(x_i',\mathbf x_{-i})).$}
Further, {$f$} 
is \emph{group strategy-proof} (GSP), if no coalition $G\subseteq N$ of agent can decrease the cost of every agent in $G$ by misreporting, regardless of the location profile $\mathbf x_{-G}$ of others, that is, for any $\mathbf x_G'$, there exists an agent $i\in G$ such that
{$c_i(f(\mathbf x_G,\mathbf x_{-G}))\le c_i(f(\mathbf x_G',\mathbf x_{-G})).$}
A mechanism $f$ is \emph{anonymous}, if for every profile $\mathbf x$ and every permutation of agents $\pi:N\rightarrow N$, it holds that
$f(x_1,\ldots,x_n)=f(x_{\pi(1)},\ldots,x_{\pi(n)})$.

Denote an instance by $I(\mathbf x,M)$ or simply $I$. We consider two objective functions, minimizing the social cost and {minimizing} the maximum cost.

\smallskip\noindent\textbf{Social cost.} Given {a} location profile $\mathbf x$, the social cost of solution $\mathbf y$ is the total distance to all agents, that is,
$$SC(\mathbf x,\mathbf y)=\sum_{i\in N}c_i(\mathbf y)=\sum_{i\in N}d(x_i,\mathbf y).$$
\smallskip\textbf{Maximum cost.} Given {a} location profile $\mathbf x$, the maximum cost of solution $\mathbf y$ is the maximum distance to all agents, that is,
$$MC(\mathbf x,\mathbf y)=\max_{i\in N}c_i(\mathbf y)=\max_{i\in N}d(x_i,\mathbf y).$$

When evaluating a mechanism's performance, we use the standard worst-case approximation notion. Formally, given an instance $I(\mathbf x,M)$, let $opt(\mathbf x)\in\arg\min_{\mathbf y\in M^k}C(\mathbf x,\mathbf y)$ be an optimal facility location profile, {and $OPT(\mathbf x)$ be the optimum value}. We say that a mechanism $f$ provides an $\alpha$-approximation if for every instance $I(\mathbf x,M)$,
$$C(\mathbf x,f(\mathbf x))\le \alpha\cdot C(\mathbf x,opt(\mathbf x))=\alpha\cdot OPT(\mathbf x),$$
where the objective function $C$ can be $SC$ or $MC$. The goal is to design deterministic or randomized strategy-proof mechanisms with small approximation ratios.

\section{{Single-facility location games}}
In this section, we study the single-facility location games, i.e., $k=1$.  Feldman \emph{et al.} \cite{feldman2016voting} thoroughly study this problem for the {social-cost} objective. When the space is a {line}, they prove tight bounds on the approximations: the Median mechanism that places the facility at the nearest candidate of the median agent is SP and $3$-approximation, and no deterministic SP mechanism can do better. They also propose a randomized SP and $2$-approximation mechanism (called the Spike mechanism) that selects the nearest candidate of each agent with specific probabilities,  and prove that this approximation ratio is {the} best possible for any randomized SP {mechanism}. When it is a general metric space, they show that random dictatorship has a best possible approximation ratio of $3-\frac2n$.

Hence, we only consider the objective of minimizing the maximum cost among the agents. We study the problem on a {line} and {in} a general metric space, respectively.

\subsection{Line space}

Suppose {that} the space is a {line}. Let $x_l$ (resp. $x_r$) be the location of the leftmost (resp. rightmost) agent with respect to location profile $\mathbf x$. Consider the following mechanism.

\begin{mechanism}\label{mech:singleL}
Given {a} location profile $\mathbf x$, select the candidate location which is closest to the leftmost agent, that is, select a candidate in {location} $\arg\min_{y\in M}|y-x_l|$, breaking ties in {any} 
deterministic way.
\end{mechanism}

\begin{theorem}\label{thm:max1}
For the single-facility problem on a line, Mechanism \ref{mech:singleL} is an anonymous GSP and 3-approximation mechanism, under the {maximum-cost} objective.
\end{theorem}
\begin{proof}
Denoted by $f$ Mechanism \ref{mech:singleL}. It is clearly anonymous, because the outcome depends only on the agent locations, not on their identities. Namely, for any permutation of the agent locations on the line, the facility locations do not change. Let $y=f(\mathbf x)$ be the outcome of the mechanism, and define $L=|x_r-x_l|$. We discuss three cases with respect to the location of $y$.

{Case 1: $x_l-L/2\le y\le x_r$.} The maximum cost of $y$ is $MC(\mathbf x,y)=\max\{|y-x_l|,|y-x_r|\}\le 3L/2$, and the optimal maximum cost is at least $L/2$. So we have $\frac{MC(\mathbf x,y)}{OPT(\mathbf x)}\le 3$.

{Case 2: $y>x_r$.} It is easy to see that $y$ is an optimal solution with maximum cost $|y-x_l|$, because the closest candidate $y'$ to the left of $x_l$ induces a maximum cost of at least $|y-x_l|+L$, and there is no {candidate} between $y'$ and $y$.

{Case 3: $y<x_l-L/2$.} The maximum cost of $y$ is $MC(\mathbf x,y)=x_r-y$. The optimal candidate has a distance at least $x_l-y$ to $x_l$. So  we have
{$$\frac{MC(\mathbf x,y)}{OPT(\mathbf x)}\le \frac{x_r-y}{x_l-y}=\frac{x_l+L-y}{x_l-y}=1+\frac{L}{x_l-y}<1+\frac{L}{L/2}=3,
$$}
which establishes the proof for approximation ratio.

It remains to show the group strategy-proofness. For any group of agents $G$, we want to show at least one agent in $G$ cannot gain by misreporting. Clearly the agent located at $x_l$ has no incentive to join $G$, because he already attains the minimum possible cost. The only way for $G$ to influence the output of mechanism $f$ is someone reporting a location to the left of $x_l$. However, this cannot move the facility location to the right, and thus no {agent} in $G$ can benefit by misreporting.
\end{proof}

Let $\epsilon>0$ be a sufficiently small number. We prove lower bounds for the approximation ratio of (deterministic and randomized) SP mechanisms, matching the upper bound in Theorem \ref{thm:max1}.

\begin{theorem}\label{thm:lb}
For the single-facility problem on a line,  no deterministic (resp. randomized) SP {mechanism} can have an approximation ratio better than 3 (resp. 2), under  the {maximum-cost} objective.
\end{theorem}
\begin{proof}
Suppose $f$ is a deterministic strategy-proof mechanism with approximation ratio $3-\delta$ for some $\delta>0$. Consider an instance $I$ (as shown in {Fig.} 
\ref{fig:33}) with agents' location profile $\mathbf x=(1-\epsilon,1+\epsilon)$, and $M=\{0,2\}$. By symmetry, assume w.l.o.g. that $f(\mathbf x)=0$. The cost of agent 2 is $c_2(0)=|1+\epsilon-0|=1+\epsilon$. Now consider another instance $I'$ with agents' location profile $\mathbf x'=(1-\epsilon,3)$, and $M=\{0,2\}$. The optimal solution is candidate $2$, and the optimal maximum cost is $1+\epsilon$. The maximum cost induced by candidate $0$ is $3$. Since the approximation ratio of $f$ is $3-\delta$ and $\epsilon\rightarrow 0$, it must select $f(\mathbf x')=2$. It indicates that, under instance $I$, agent 2 located at $x_2=1+\epsilon$ can decrease his cost from $c_2(0)$ to $c_2(f(\mathbf x'))=|1+\epsilon-2|=1-\epsilon$, by misreporting his location as $x_2'=3$. This is a contradiction with  {strategy-proofness.}
\begin{figure}[htpb]
\begin{center}
\includegraphics[scale=0.32]{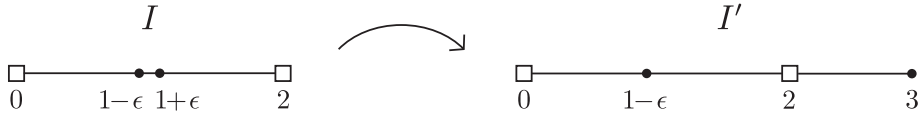}
\caption{\label{fig:33}  Two instances $I$ and $I'$, where hollow squares indicate candidates, and solid circles indicate agents. }
\end{center}
\end{figure}

Suppose $f$ is a randomized strategy-proof mechanism with approximation ratio $2-\delta$ for some $\delta>0$. Also consider instance $I$. W.l.o.g. assume that $f(\mathbf x)=0$ with probability at least $\frac12$. The cost of agent 2 is $c_2(f(\mathbf x))\ge \frac12(1+\epsilon)+\frac12(1-\epsilon)=1$. Then consider instance $I'$.  Let $P'_2$ be the probability of $f(\mathbf x')=2$.  Since the approximation ratio of $f$ is $2-\delta$, we have
$$ P'_2\cdot \frac{MC(\mathbf x',2)}{OPT(\mathbf x')}+(1-P'_2)\cdot\frac{MC(\mathbf x',0)}{OPT(\mathbf x')}=P'_2+(1-P'_2)\cdot\frac{3}{1+\epsilon}\le 2-\delta,$$
which implies that $P'_2>\frac12$ as $\epsilon\rightarrow 0$.
Hence, under instance $I$, agent 2 located at $x_2=1+\epsilon$ can decrease his cost to $c_2(f(\mathbf x'))< \frac12(1-\epsilon)+\frac12(1+\epsilon)=1\le c_2(f(\mathbf x))$, by misreporting his location as $x_2'=3$.
\end{proof}

\begin{remark}
For randomized mechanisms, we have shown {that} the lower bound is 2, and {we} are failed to find a matching upper bound. {We are concerned with \emph{weighted percentile voting} (WPV) mechanisms (see \cite{feldman2016voting}), which {locate} the facility on the $i$-th percentile agent's closest candidate with some probability $p_i$, where $p_i$ does not depend on the location profile $\mathbf x$. For example, Mechanism \ref{mech:singleL} is WPV by setting $p_0=1$ and $p_i=0$ for $i>0$.}  We remark that no WPV {mechanism} can beat the ratio of 3. Consider an instance with agents' location profile $(1,3)$ and candidates' location profile $(\epsilon,2,4-\epsilon)$. The optimal maximum cost is 1, attained by selecting candidate 2, while any WPV mechanism must select either candidate $\epsilon$ or {$4-\epsilon$}, 
inducing a maximum cost of $3-\epsilon$. The ratio approaches 3 when {$\epsilon$ tends to $0$}. It leaves an open question to narrow this gap.
\end{remark}

\subsection{General metric spaces}

In this subsection, we extend the model from a {line} to a general metric space. In this setting, the locations of  all agents and facility candidates are {in} a metric space $(S,d)$. Our objective is to  minimize the maximum cost of agents. We give the following dictatorial mechanism, in which the dictator can be an arbitrary agent.

\begin{mechanism}[{Dictatorship}]\label{mech:singleM}
Given {a} location profile $\mathbf x$, for an arbitrary agent $k\in N$, select the closest candidate location to agent $k$, that is, $\arg\min_{y\in M}d(x_k,y)$, breaking ties in {any} 
deterministic way.
\end{mechanism}

\begin{theorem} \label{thm:gsp}
For the single-facility problem {in} a metric space, Mechanism \ref{mech:singleM} is  GSP and 3-approximation, under the {maximum-cost} objective.
\end{theorem}
\begin{proof}
Denote $f$ by Mechanism \ref{mech:singleM}. Let $y=f_k(\mathbf x)$ for a fixed $k$, and $y^*=opt(\mathbf x)$ be the optimal solution. 
Then we have $d(x_i,y^*) \leq OPT(\mathbf x)$ for each agent $i\in N$. As the distance function has the triangle inequality property in a metric space, we derive the following for each $i\in N$:
\begin{flalign}
\begin{split}
d(x_i, y) &\leq d(y^*, y) + d(x_i, y^*)\\
&\leq d(x_k, y) + d(x_k, y^*) + d(x_i, y^*)\\
&\leq 2d(x_k, y^*) + d(x_i, y^*)\\
&\leq 3~OPT,
\end{split}
\end{flalign}

The group strategy-proofness is trivial, because Mechanism \ref{mech:singleM} is dictatorial. 
\end{proof}

{One can find that, though losing the anonymity, Mechanism \ref{mech:singleM} is indeed a generalization of Mechanism \ref{mech:singleL}, and the approximation ratio in Theorem \ref{thm:gsp} implies that in Theorem \ref{thm:max1}.  }

Recall that random dictatorship locates the facility on agent $i$'s closest candidate with probability $1/n$ for all $i\in N$.
It has an approximation ratio of
$3-\frac2n$ for the {social-cost} objective {in} any metric space \cite{anshelevich2017randomized}. However, it does not help to improve the deterministic upper bound 3 in Theorem \ref{thm:gsp}, even if on the line.

\section{Two-facility location games}
{In this section, we consider the two-facility location games on a {line}, under both objectives of minimizing the social cost and {minimizing} the maximum cost. We give a linear approximation for the {social-cost} objective, which asymptotically tight, and a 3-approximation for the {maximum-cost} objective, which is best possible. }

\subsection{Social-cost objective} 

For the classic (unconstrained) facility location games in a continuous line under the social-cost objective, Fotakis and Tzamos \cite{fotakis2014power} prove that no deterministic mechanism has an approximation ratio less than $n-2$. Note that the lower bound $n-2$ also holds in our setting, because when all points on the line are candidates, our problem is equivalent to the classic problem. For the same setting in \cite{fotakis2014power}, Procaccia \emph{et al.} \cite{procaccia2009approximate} give a GSP $(n$-$2)$-approximation mechanism, which selects the two extreme agent locations.  We generalize this mechanism to our setting.

\begin{mechanism}\label{mech:twoL}
Given {a} location profile $\mathbf x$ on a {line}, select the candidate location which is closest to the leftmost agent, (i.e., $\arg\min_{y\in M}|y-x_l|$), {breaking ties in favor of the candidate to the right;} and select the one closest to the rightmost agent (i.e., $\arg\min_{y\in M}|y-x_r|$), {breaking ties in favor of the candidate to the left.}
\end{mechanism}

\begin{lemma}\label{lem:m4sp}
Mechanism \ref{mech:twoL} is GSP.
\end{lemma}

\begin{proof}
For any group $G$ of agents, we want to show at least one agent in $G$ cannot gain by misreporting. Clearly the agent located at $x_l$ or $x_r$ has no incentive to join the coalition $G$, because he already attains the minimum possible cost. The only way for $G$ to influence the output of the mechanism is some member reporting a location to the left of $x_l$ or the right of $x_r$. However, {this can move neither of the two facility locations closer to the members.} So no {agent} in $G$ can benefit by misreporting.
\end{proof}

\begin{theorem}\label{thm:ee}
For the two-facility problem on a line, Mechanism \ref{mech:twoL} is GSP, anonymous,  and $(2n-3)$-approximation under the {social-cost} objective.
\end{theorem}
\begin{proof}
The group strategy-proofness is given in Lemma \ref{lem:m4sp}.
Let $\mathbf y^*=(y^*_1,y^*_2)$ be an optimal solution with $y^*_1\le y^*_2$, and $\mathbf y=(y_1,y_2)$ with $y_1\le y_2$ be the solution output by Mechanism \ref{mech:twoL}. Let $N_1=\{i\in N|d(x_i,y^*_1)\le d(x_i,y^*_2)\}$ be the set of agents who are closer to $y_1^*$ in the optimal solution, and $N_2=\{i\in N|d(x_i,y^*_1)> d(x_i,y^*_2)\}$ be the complement set. {Renaming if necessary, we assume $x_l=x_1\le\cdots \le x_n=x_r$.  } {If $|N_1|=0$, then $y_2^*$ must be the closest candidate to every agent in $N$, including $x_1$ (i.e., $y_2^*\in \arg\min_{y\in M}d(x_1,y)$). By the specific way of tie-breaking, Mechanism \ref{mech:twoL} must select $y_2^*$, and achieve the optimality. The symmetric analysis holds for the case when $|N_2|=0$. }

{So we only need to consider the case when $|N_1|\geq 1$ and $|N_2|\geq 1$. Clearly, $1\in N_1$ and $n\in N_2$.
} The social cost of the outcome $\mathbf y$ by Mechanism \ref{mech:twoL} is
\begin{flalign*}
&\sum_{i\in N} \min\{d(x_i,y_1),d(x_i,y_2)\}\leq \sum_{i\in N_1}d(x_i,y_1)+\sum_{i\in N_2}d(x_i,y_2)\\
\leq &\sum_{i\in N_1\backslash \{1\}}[~d(x_i,y^*_1)+d(y^*_1,y_1)~]+\sum_{i\in N_2\backslash \{n\}}[~ d(x_i,y_2^*)+d(y^*_2,y_2)~]\\
& + d(x_l,y^*_1) + d(x_r,y^*_2)\\
=& ~OPT+ \sum_{i\in N_1\backslash \{1\}}d(y^*_1,y_1)+\sum_{i\in N_2\backslash \{n\}}d(y^*_2,y_2)\\
\leq & ~OPT+ \sum_{i\in N_1\backslash \{1\}}[~d(x_l,y^*_1)+d(x_l,y_1)~]+\sum_{i\in N_2\backslash \{n\}}[~ d(x_r,y_2^*)+d(x_r,y_2)~]\\
\leq & ~OPT+ 2(|N_1|-1)\cdot d(x_l,y^*_1) + 2(|N_2|-1)\cdot d(x_r,y^*_2)\\
\le & ~OPT+2(\max\{|N_1|,|N_2|\}-1)\cdot OPT\\
\leq & ~(2n-3)\cdot OPT,
\end{flalign*}
{where the second last inequality {holds} because  $OPT\ge d(x_l,y^*_1)+d(x_r,y^*_2)$, and the last inequality {holds} because $\max\{|N_1|,|N_2|\}\le n-1$. }
\end{proof}

{Next} we give an example to show that the analysis in Theorem \ref{thm:ee} for the approximation ratio of Mechanism \ref{mech:twoL} is tight.

\begin{example}
Consider an instance on a line with agents' location profile $\mathbf x=(1,\frac43,\ldots,\frac43,2)$ and candidates' location profile $(\frac23+\epsilon,\frac43,2)$. The optimal social cost is $\frac13$, attained by solution $(\frac43,2)$. Mechanism \ref{mech:twoL} outputs solution $(\frac23+\epsilon,2)$, and the social cost is $(\frac43-\frac23-\epsilon)\cdot (n-2)+\frac13-\epsilon$. Then we have $\frac{(2/3-\epsilon)\cdot (n-2)+1/3-\epsilon}{1/3}\rightarrow 2n-3$, when $\epsilon$ tends to 0.
\end{example}

\subsection{{Maximum-cost objective}}

Next, we turn to consider the {maximum-cost} objective.

\begin{theorem}\label{thm:m4}
For the two-facility problem on a line, Mechanism \ref{mech:twoL} is GSP, anonymous, and 3-approximation under the {maximum-cost} objective.
\end{theorem}
\begin{proof}
The group strategy-proofness is given in Lemma \ref{lem:m4sp}. For any location profile $\mathbf x$, let $\mathbf y^*=(y^*_1,y^*_2)$ be an optimal solution with $y^*_1\le y^*_2$, and $\mathbf y=(y_1,y_2)$ with $y_1\le y_2$ be the solution output by Mechanism \ref{mech:twoL}.  Assume w.l.o.g. that $x_1\le\cdots\le x_n$. Let $N_1=\{i\in N|d(x_i,y^*_1)\le d(x_i,y^*_2)\}$ be the set of agents who are closer to $y_1^*$ in the optimal solution, and $N_2=\{i\in N|d(x_i,y^*_1)> d(x_i,y^*_2)\}$ be the complement set. {Let $n_1=|N_1|$ and $n_2=|N_2|$.} Define $C_1=\max_{i\in N_1}d(x_i,y^*_1)$ and $C_2= \max_{i\in N_2}d(x_i,y^*_2)$. It is easy to see that the optimal maximum cost is $\max\{C_1,C_2\}$.

 Next, we consider {a} restricted instance $(x_1,\ldots,x_{n_1})$ of the single-facility location problem. By the definition of Mechanism \ref{mech:twoL}, candidate $y_1$ is the closest one to agent 1. By Theorem \ref{thm:max1}, we have $\max_{i\in N_1}d(y_1,x_i)\le 3C_1$. Similarly, consider another restricted instance $(x_{n_1+1},\ldots,x_{n})$, we have $\max_{i\in N_2}d(y_2,x_i)\le 3C_2$. Therefore,
 \[\max_{i\in N}d(x_i,\mathbf y)\le 3\max\{C_1,C_2\},\]
 which completes the proof.

\end{proof}

In the following we give a lower bound 2 for randomized SP mechanisms, and a lower bound 3 for deterministic SP mechanisms, matching the bound in Theorem \ref{thm:m4}. We use the same construction as in the proof of Theorem \ref{thm:lb} for $2$ agents, and locate an additional
agent at a very far away point in all the location profiles used in the proof.

\begin{theorem}\label{thm:6}
For the two-facility problem on a line, no deterministic (resp. randomized) SP {mechanism} can have an approximation ratio better
than 3 (resp. 2), under the {maximum-cost} objective.
\end{theorem}
\begin{proof}
Suppose $f$ is a deterministic SP mechanism with approximation ratio $3-\delta$ for some $\delta>0$. Consider an instance $I$ (as shown in Figure \ref{fig:3}) with agents' location profile $\mathbf x=(1-\epsilon,1+\epsilon,L)$, and $M=\{0,2,L\}$, where $L$ is sufficiently large and $\epsilon>0$ is sufficiently small. Note that candidate $L$ must be selected (to serve agent 3) by any {mechanism that has a} good approximation {ratio}. 
We can assume w.l.o.g. that $0\in f(\mathbf x)$. The cost of agent 2 is $c_2(f(\mathbf x))=|1+\epsilon-0|=1+\epsilon$. Now consider another instance $I'$ with agents' location profile $\mathbf x'=(1-\epsilon,3,L)$, and $M=\{0,2,L\}$.  The optimal maximum cost is $1+\epsilon$, attained by selecting candidates $2$ and $L$. The maximum cost induced by any solution that selects candidate $0$ is at least $3$. Since the approximation ratio of $f$ is $3-\delta$ and $\epsilon\rightarrow 0$, it must select candidate 2, i.e., $2\in f(\mathbf x')$. It indicates that, under instance $I$, agent 2 located at $1+\epsilon$ can decrease his cost from $c_2(f(\mathbf x))=1+\epsilon$ to $c_2(f(\mathbf x'))=|1+\epsilon-2|=1-\epsilon$, by misreporting his location as $x_2'=3$. This is a contradiction with the strategy-proofness.
\begin{figure}[htpb]
\begin{center}
\includegraphics[scale=0.265]{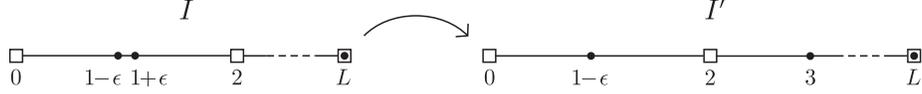}
\caption{\label{fig:3} There is an agent and a candidate in a very far away location $L$.}
\end{center}
\end{figure}

Suppose $f$ is a randomized SP mechanism with approximation ratio $2-\delta$ for some $\delta>0$. Also consider instance $I$. Note that candidate $L$ must be selected with probability 1 by any mechanisms that have {a good approximation ratio, since $L$ tends to $\infty$}. 
We can assume w.l.o.g. that $0\in f(\mathbf x)$ with probability at least $\frac12$. The cost of agent 2 is $c_2(f(\mathbf x))\ge \frac12(1+\epsilon)+\frac12(1-\epsilon)=1$. Then consider instance $I'$. The optimal maximum cost is $1+\epsilon$, which is attained by selecting candidates 2 and $L$. Let $P'_2$ be the probability of $f$ selecting candidate 2. The maximum cost induced by any solution that selects candidate $0$ is at least $3$.  Since the approximation ratio of $f$ is $2-\delta$, we have
 $$P'_2\cdot 1+(1-P'_2)\cdot\frac{3}{1+\epsilon}\le2-\delta,$$
 which implies that $P'_2>\frac12$ as {$\epsilon\rightarrow0$}. Hence, under instance $I$, agent 2 located at $1+\epsilon$ can decrease his cost to $c_2(f(\mathbf x'))< \frac12(1-\epsilon)+\frac12(1+\epsilon)\le c_2(f(\mathbf x))$, by misreporting his location as $x_2'=3$.
\end{proof}

\section{Conclusion}

For the classic $k$-facility location games, Fotakis and Tzamos \cite{fotakis2014power} show that for every $k\ge 3$, there do not exist
any deterministic anonymous SP mechanisms with a bounded approximation ratio for the {social-cost} objective on the line, even for simple instances with $k+1$ agents. It directly follows a corollary that there  {exists no} such mechanism with a bounded approximation ratio for the {maximum-cost} objective. Therefore, in our constrained setting with candidate locations, we cannot expect to beat such lower bounds when $k\ge 3$.

In this paper we are concerned with designing truthful deterministic mechanisms for the setting with candidates. It remains an open question to find randomized mechanisms matching the lower bound 2 in Theorems \ref{thm:lb} and \ref{thm:6}, though we have excluded the possibility of WPV mechanisms.

\section*{Acknowledge}
The authors thank Minming Li and three anonymous referees for their invaluable suggestions and comments.
Minming Li is supported by NNSF of China under Grant No. 11771365, and sponsored by Project No. CityU 11205619 from Research Grants Council of HKSAR.

\bibliography{reference}

\begin{thebibliography}{10}

\bibitem{alon2010strategyproof}
Noga Alon, Michal Feldman, Ariel~D Procaccia, and Moshe Tennenholtz.
\newblock Strategyproof approximation of the minimax on networks.
\newblock {\em Mathematics of Operations Research}, 35(3):513--526, 2010.

\bibitem{anshelevich2017randomized}
Elliot Anshelevich and John Postl.
\newblock Randomized social choice functions under metric preferences.
\newblock {\em Journal of Artificial Intelligence Research}, 58:797--827, 2017.

\bibitem{Aziz2019TheCC}
Haris Aziz, Hau Chan, Barton~E. Lee, and David~C. Parkes.
\newblock The capacity constrained facility location problem.
\newblock In {\em Proceeding of the 15th Conference on Web and Internet
  Economics (WINE)}, page 336, 2019.

\bibitem{border1983straightforward}
Kim~C Border and James~S Jordan.
\newblock Straightforward elections, unanimity and phantom voters.
\newblock {\em The Review of Economic Studies}, 50(1):153--170, 1983.

\bibitem{chen2018mechanism}
Xujin Chen, Xiaodong Hu, Xiaohua Jia, Minming Li, Zhongzheng Tang, and Chenhao
  Wang.
\newblock Mechanism design for two-opposite-facility location games with
  penalties on distance.
\newblock In {\em Proceedings of the 11th International Symposium on
  Algorithmic Game Theory (SAGT)}, pages 256--260, 2018.

\bibitem{chen2019truthful}
Xujin Chen, Minming Li, Changjun Wang, Chenhao Wang, and Yingchao Zhao.
\newblock Truthful mechanisms for location games of dual-role facilities.
\newblock In {\em Proceedings of the 18th International Conference on
  Autonomous Agents and MultiAgent Systems (AAMAS)}, pages 1470--1478, 2019.

\bibitem{dokow2012mechanism}
Elad Dokow, Michal Feldman, Reshef Meir, and Ilan Nehama.
\newblock Mechanism design on discrete lines and cycles.
\newblock In {\em Proceedings of the 13th ACM Conference on Electronic Commerce
  (ACM-EC)}, pages 423--440, 2012.

\bibitem{feldman2016voting}
Michal Feldman, Amos Fiat, and Iddan Golomb.
\newblock On voting and facility location.
\newblock In {\em Proceedings of the 17th ACM Conference on Economics and
  Computation (ACM-EC)}, pages 269--286, 2016.

\bibitem{fotakis2014power}
Dimitris Fotakis and Christos Tzamos.
\newblock On the power of deterministic mechanisms for facility location games.
\newblock {\em ACM Transactions on Economics and Computation (TEAC)},
  2(4):1--37, 2014.

\bibitem{kyropoulou2019mechanism}
Maria Kyropoulou, Carmine Ventre, and Xiaomeng Zhang.
\newblock Mechanism design for constrained heterogeneous facility location.
\newblock In {\em Proceedings of the 12th International Symposium on
  Algorithmic Game Theory (SAGT)}, pages 63--76, 2019.

\bibitem{libudgeted}
Minming Li, Chenhao Wang, and Mengqi Zhang.
\newblock Budgeted facility location games with strategic facilities.
\newblock In {\em Proceedings of the 29th International Joint Conference on
  Artificial Intelligence (IJCAI)}, pages 400--406, 2020.

\bibitem{lu2010asymptotically}
Pinyan Lu, Xiaorui Sun, Yajun Wang, and Zeyuan~Allen Zhu.
\newblock Asymptotically optimal strategy-proof mechanisms for two-facility
  games.
\newblock In {\em Proceedings of the 11th ACM Conference on Electronic Commerce
  (ACM-EC)}, pages 315--324, 2010.

\bibitem{lu2009tighter}
Pinyan Lu, Yajun Wang, and Yuan Zhou.
\newblock Tighter bounds for facility games.
\newblock In {\em Proceedings of the 5th International Workshop on Internet and
  Network Economics (WINE)}, pages 137--148, 2009.

\bibitem{meir2018strategic}
Reshef Meir.
\newblock Strategic voting.
\newblock {\em Synthesis Lectures on Artificial Intelligence and Machine
  Learning}, 13(1):1--167, 2018.

\bibitem{meir2012algorithms}
Reshef Meir, Ariel~D Procaccia, and Jeffrey~S Rosenschein.
\newblock Algorithms for strategyproof classification.
\newblock {\em Artificial Intelligence}, 186:123--156, 2012.

\bibitem{procaccia2009approximate}
Ariel~D. Procaccia and Moshe Tennenholtz.
\newblock Approximate mechanism design without money.
\newblock In {\em Proceedings of the 10th ACM Conference on Electronic Commerce
  (ACM-EC)}, pages 177--186, 2009.

\bibitem{procaccia2013approximate}
Ariel~D Procaccia and Moshe Tennenholtz.
\newblock Approximate mechanism design without money.
\newblock {\em ACM Transactions on Economics and Computation (TEAC)},
  1(4):1--26, 2013.

\bibitem{schummer2002strategy}
James Schummer and Rakesh~V Vohra.
\newblock Strategy-proof location on a network.
\newblock {\em Journal of Economic Theory}, 104(2):405--428, 2002.

\bibitem{serafino2015truthful}
Paolo Serafino and Carmine Ventre.
\newblock Truthful mechanisms without money for non-utilitarian heterogeneous
  facility location.
\newblock In {\em Proceedings of the 29th AAAI Conference on Artificial
  Intelligence (AAAI)}, pages 1029--1035, 2015.

\bibitem{sui2015approximately}
Xin Sui and Craig Boutilier.
\newblock Approximately strategy-proof mechanisms for (constrained) facility
  location.
\newblock In {\em Proceedings of the 14th International Conference on
  Autonomous Agents and Multiagent Systems (AAMAS)}, pages 605--613, 2015.

\bibitem{walsh2020strategy}
Toby Walsh.
\newblock Strategy proof mechanisms for facility location at limited locations.
\newblock {\em arXiv preprint arXiv:2009.07982}, 2020.

\bibitem{zou2015facility}
Shaokun Zou and Minming Li.
\newblock Facility location games with dual preference.
\newblock In {\em Proceedings of the 14th International Conference on
  Autonomous Agents and Multiagent Systems (AAMAS)}, pages 615--623, 2015.

\end{thebibliography}

\end{document}